\documentclass[11pt]{elsarticle}
\usepackage{setspace}
\usepackage[plain]{fullpage}
\usepackage{algorithm}
\usepackage{algorithmicx}
\usepackage[noend]{algpseudocode}
\usepackage{algpascal}
\usepackage{graphicx}

\usepackage{graphics}
\usepackage{epsfig}
\usepackage{enumerate}
\usepackage{color}
\usepackage{amsmath, amssymb, amsthm}

\newtheorem{thm}{Theorem}[section]
\newtheorem{prop}[thm]{Proposition}
\newtheorem{lem}[thm]{Lemma}
\newtheorem{cor}[thm]{Corollary}

\newtheorem{rmk}[thm]{Remark}
\newtheorem{claim}[thm]{Claim}

\journal{}

\begin{document}
\begin{frontmatter}
\title{Complexity of Improper  Twin Edge Coloring of Graphs\tnoteref{MAD}}
\tnotetext[MAD]{Part of this research was carried out while Saieed Akbari was visiting Istanbul Center for Mathematical Sciences (IMBM) whose support is greatly acknowledged.}

\author[sharif]{Paniz Abedin}

\author[sharif]{Saieed Akbari}

%
\author[rmit]{Marc Demange}
\author[bu]{T{\i}naz Ekim\corref{corr}}
%
%
%
\cortext[corr]{Corresponding Author}
\address[sharif]{Department of Mathematical Sciences, Sharif University of Technology, 11155-9415, Tehran, Iran}
\address[rmit]{School of Science, RMIT University, Melbourne, Victoria, Australia}
\address[bu]{Department of Industrial Engineering, Bogazici University, Istanbul, Turkey}

\begin{abstract}
 Let $G$ be a graph whose each component has order at least 3. Let $s : E(G) \rightarrow \mathbb{Z}_k$ for some integer $k\geq 2$ be an improper edge coloring of $G$ (where adjacent edges may be assigned the same color). If the induced vertex coloring $c : V (G) \rightarrow \mathbb{Z}_k$ defined by $c(v) = \sum_{e\in E_v} s(e) \mbox{ in } \mathbb{Z}_k,$ (where the indicated sum is computed in $\mathbb{Z}_k$ and $E_v$ denotes the set of all edges incident to $v$) results in a proper vertex coloring
of $G$, then we refer to such a coloring as an  improper twin $k$-edge coloring. The minimum $k$ for which $G$ has an improper twin $k$-edge coloring is called
the improper twin chromatic index of $G$ and is denoted by $\chi'_{it}(G)$.\\
It is known that $\chi'_{it}(G)=\chi(G)$, unless $\chi(G)=2 \pmod 4$
and in this case  $\chi'_{it}(G)\in \{\chi(G), \chi(G)+1\}$. In this paper, we first give a short proof of this result and we show that if $G$ admits an improper twin $k$-edge coloring for some positive integer $k$, then $G$ admits an improper twin $t$-edge coloring for all $t\geq k$; we call this the monotonicity property. In addition, we provide a linear time algorithm to construct an improper twin edge coloring using at most $k+1$ colors, whenever a $k$-vertex coloring is given. Then we investigate, to the best of our knowledge the first time in literature, the complexity of deciding whether $\chi'_{it}(G)=\chi(G)$ or
$\chi'_{it}(G)=\chi(G)+1$, and  we show that, just like in case of the edge chromatic index, it is NP-hard even in some restricted cases. Lastly, we exhibit several classes of graphs for which the problem is polynomial.

\end{abstract}


\end{frontmatter}

\section{Introduction}\label{sec:intro}

\subsection{Definitions and Literature}

Throughout this paper all graphs are simple with no loops and multiple edges.
Let $G$ be a graph. We denote the vertex set and the edge set of $G$, by $V(G)$ and $E(G)$, respectively (or simply $V$ and $E$ when no ambiguity may occur). For $v\in V(G)$, $E_v$ denotes the set of
 all edges incident to $v$ and $N_G(v)$ denotes the (open) neighborhood of $v$ in $G$; when no ambiguity may occur we just denote $N(v)=N_G(v)$. We denote respectively by $K_n$, $P_n$ and $C_n$ the complete graph of order $n$, the path or order $n$ and the cycle of order $n$. A {\it $k$-clique} is a set of $k$ vertices inducing a complete graph and an independent set is a set of vertices pairwise not linked by an edge (i.e., inducing a subgraph with only isolated vertices). $K_{p,q}$ will denote the complete bipartite graph with parts of size $p$ and $q$. If $p=1$, then it is  star graph and the vertex of the part of size 1 is called  center.
 
 A {\it (proper) $k$-vertex coloring} of $G$ is an assignment of $k$ colors to the vertices of $G$ such that no two adjacent vertices have the same color (possibly some colors are not used). 
 Given such a $k$-vertex coloring, a color class is the set of vertices of the same color; it is an independent set. A graph having a $k$-vertex coloring is said $k$-vertex colorable. 
 The {\it chromatic number} $\chi(G)$ of $G$ is the minimum $k$ for which $G$ is $k$-vertex colorable.  
 
 A {\it $k$-edge coloring} of $G$ is an assignment of $k$ colors to the edges of $G$ (possibly some colors are not used); it is said to be {\it proper} if two edges with a common endpoint do not have the same color, and {\it improper} if it is not necessarily proper. Note that the distinction between proper and improper colorings exists as well for vertex colorings but since we will not use it (formally since all our vertex colorings are proper) we specify proper/improper  for edge colorings only. The minimum $k$ such that $G$ admits a proper $k$-edge coloring is the {\it chromatic index} of $G$  denoted by $\chi'(G)$. Integers will be used to label colors, which will give us the opportunity to compute operations.
   
 A {\it walk} in $G$ is a  sequence $W=v_0e_1v_1e_2v_2 \ldots, e_kv_k$, whose terms are alternately vertices and edges, such that $e_i=v_{i-1}v_i$, for $1\leq i \leq k$. The number $k$ is called the length of the walk. A walk is called a {\it closed walk} if $v_0=v_k$. 
The reader is referred to~\cite{golumbic} for all definitions in graph theory which are not defined here.

For any integer $k\geq 2$, we will denote by $\mathbb{Z}_k= \mathbb{Z}/k\mathbb{Z}$ the additive group over $\{0, \ldots, k-1\}$. For $x,y,z$ in  $\mathbb{Z}_k$, $x+y=z$ in $\mathbb{Z}_k$ is also denoted $z=x+y \mod k$. For any two integers, we will denote $x\equiv y \pmod k$ if $0= (x-y) \mod k$.

In the literature, we can find a number of results about proper or improper edge colorings {\it inducing} a coloring for vertices by mapping, for any vertex $v$, the multiset of colors of the edges 
adjacent to $v$ to a vertex color. Such a mapping is called \textit{vertex-coloring mapping}; see \cite{colorinduced} for variations of such colorings. Given a vertex-coloring mapping, an edge-coloring is called \textit{vertex-distinguishing} if the colors induced by the edge coloring are different for every vertex. The edge-coloring is called \textit{neighbor-distinguishing} if only adjacent vertices have different induced colors, see e.g.  \cite{ vertexcol1} and \cite{vertexcol2}. Neighbor-distinguishing edge colorings are also called \textit{vertex coloring edge partitions} in \cite{vertexcol2}. Chapter 13 of \cite{chromaticGT} is devoted to distinguishing colorings. 

These notions strongly vary with respect to  the vertex coloring mapping, i.e., the way the vertex colors are induced from the edge colors. For instance, in \cite{vertexcol2}, the color of the vertex $v$ is directly the multiset of the colors in $E_v$ while in \cite{vertexcol1}, the color of a vertex is obtained as the sum of the colors of the edges incident to it and the related parameter is called \textit{sum chromatic number}. This definition gives rise to the ``1-2-3 Conjecture'' which states that for every connected graph $G$ of order at least 3, there exists a neighbor-distinguishing edge coloring of $G$ using only 3 colors. 
A variation of the sum chromatic number obtained using Abelian groups is called the \textit{group sum chromatic number} and studied in \cite{groupsum}; it is the smallest value $s$ such that taking any Abelian group $\mathcal{G}$ of
order $s$, there exists a function $f : E(G) \longrightarrow \mathcal{G}$ such that the sums of edge colors properly color the vertices i.e., with different colors on adjacent vertices.

Another popular vertex coloring mapping uses the specific group $\mathbb{Z}_k$, where $k$ is at least the number of colors used for the edges. If the edges are colored using colors between 0 and $k-1$, seen as elements of $\mathbb{Z}_k$, then one assigns to a vertex $v$  the sum in $\mathbb{Z}_k$ of the colors of the edges in $E_v$. More formally, let $G$ be a  graph the components of which have order at least 3. Let  $s : E(G) \rightarrow \mathbb{Z}_k$ for some integer $k\geq 2$ be a  $k$-edge coloring of $G$. We then define  $c_s : V (G) \rightarrow \mathbb{Z}_k$  as follows:
$$ \forall v\in V, c_s(v) = \sum_{e\in E_v} s(e) \mbox{ in } \mathbb{Z}_k.$$ If $c_s$ is a (proper) $k$-vertex coloring, then it will be called the {\it vertex coloring induced by $s$} and $s$ is said to {\it induce} $c_s$. 


A  proper edge coloring $s$ inducing a vertex coloring $c_s$ is called a {\it proper twin $k$-edge coloring} of $G$. This concept was defined in \cite{and} where it has been shown that any graph with components of size at least 3 has at least one such coloring. Then, the minimum $k$ for which $G$ has a proper twin $k$-edge coloring is called the {\it proper twin chromatic index of} $G$  and is denoted by $\chi'_{t}(G)$.

Similarly, an improper (i.e. not necessarily proper) edge coloring $s$ inducing a vertex coloring $c_s$ is called an {\it improper twin $k$-edge coloring} of $G$. The difference between both notions is only whether one deals with proper or improper edge colorings, but the related vertex coloring is always required to be proper. Since a proper edge coloring is improper as well, the existence of improper twin $k$-edge colorings for some $k$ comes from the existence of proper ones. The minimum $k$ for which $G$ has an improper twin $k$-edge coloring is called the {\it improper twin chromatic index} of $G$  and is denoted by $\chi'_{it}(G)$. The same parameter was already studied under the name of \textit{modular chromatic index}, denoted by $\chi_m'(G)$, in \cite{JCMCC11} and \cite{JCMCC12}\footnote{We 
use the term ``improper twin chromatic index'' following the terminology in \cite{and} since it allows us to distinguish between the proper and the improper cases.}.

 Note that $\chi_{it}'(G)$ and $\chi'_{t}(G)$ are not defined for graphs having a component which is $K_2$. Following the terminology of~\cite{vertexcol2} we call \textit{nice} a graph without a component of size~2. However, the definitions are valid for graphs with isolated vertices by considering that an empty sum is null and then defining $c_s(v)=0$ for all isolated vertices $v$ and all twin $k$-edge colorings, $k\geq 2$.  All the graphs considered from now on will be supposed to be nice.
 By definition, for every nice graph $G$ we have: $\chi'_{t}(G)\geq \chi'_{it}(G)\geq \chi(G)$.

In this paper, we deal with 
improper twin $k$-edge colorings 
 of $G$. 
   Note that although they are very close, the group sum chromatic number and the improper twin chromatic index are not the same. In fact, for a graph $G$,  $ \chi'_{it}(G)$ does not exceed the  group sum chromatic number of $G$. Here is an example where the group sum chromatic number is greater than the improper twin chromatic index. Consider the disjoint union of $K_4$ and $K_2$ and make one vertex of $K_4$ adjacent to one vertex of $K_2$ and
call the resultant graph by $G$. Note that $G$ has order 6 with degree sequence: 4,3,3,3,2,1. The group sum chromatic number of $G$ is 5 since this graph is ugly (see Definition 1.1 of \cite{groupsum}) and ugly graphs have group sum chromatic number of $\chi(G)+1$ (see Theorem 1.2 of \cite{groupsum}). However, the improper twin chromatic index of $G$ is 4 by Theorem \ref{thm:mainfromJCMCC} reminded below.
Let us conclude the definitions with a concrete example of improper twin edge coloring. Let $P$ be the Petersen graph. Since $P$ contains an odd cycle, $\chi(P)\geq 3$. The assignment of numbers to $E(P)$ given in Figure \ref{example} shows that $\chi'_{it}(P)=3$.

\begin{figure}[htbp]
\begin{center}
\includegraphics[scale=.17]{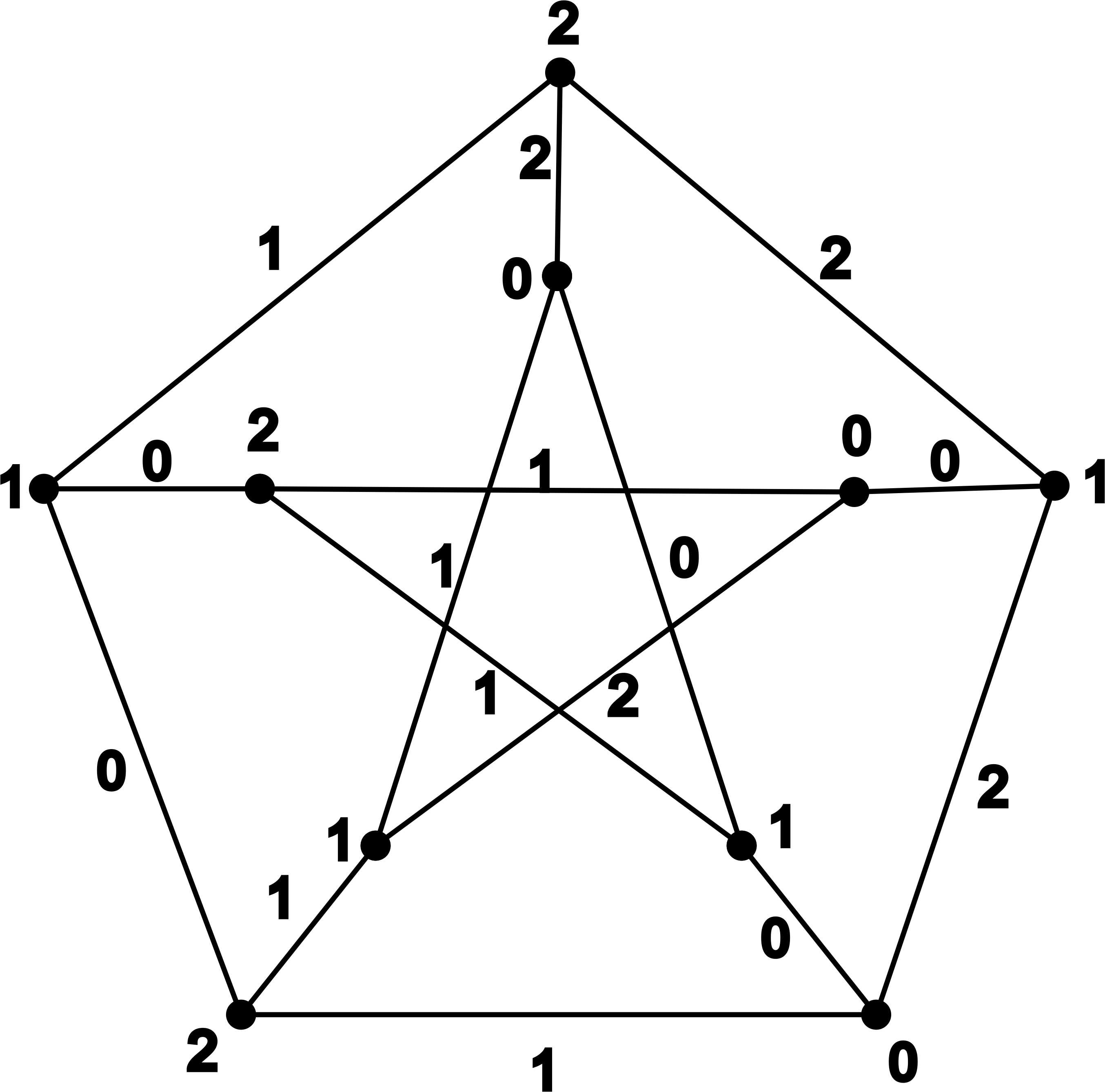}
\caption{Improper twin 3-edge coloring of the Petersen graph.}\label{example}
\end{center}
\end{figure}

\subsection{Preliminaries}
Let us now first underline that some natural properties of the classical edge and vertex colorings do not hold for the improper twin edge coloring problem.


\begin{rmk}\label{rem0}
The notions of improper twin edge coloring and induced vertex coloring strongly depend on the labels assigned to the colors. In particular they are not stable under a reordering of the colors.  
\end{rmk}
Consider indeed the star graph $K_{1,2}$; the 2-vertex coloring assigning 0 to the center and 1 to the two other vertices is induced by the edge coloring assigning 1 to both edges while the 2-vertex coloring assigning 1 to the center and 0 to the two other vertices is not induced by any improper twin edge coloring. Actually this 2-edge coloring is the unique improper 2-edge coloring. Similarly, considering a $P_4$, the edge coloring assigning color 1 to two adjacent edges and 0 to the third edge is an improper 2-edge coloring while the edge coloring assigning 0 to two adjacent edges and 1 to the last edge is not. 
\begin{rmk}\label{rem1}
Unlike the usual vertex/edge coloring problems, there are some graphs $G$ with a subgraph $H$ such that $\chi'_{it}(G)< \chi'_{it}(H)$. 
\end{rmk}

Indeed, using Theorem~\ref{thm:mainfromJCMCC} (see below), we consider $m = 2 \mod 4$ and $n = 0 \mod 4$ with $n>m>2$; then $\chi'_{it}(P_m)=3$ and $\chi'_{it}(C_n)=2$. Taking $G=C_n$ and $H=P_m$ raises infinitely many examples where taking a subgraph increases the improper twin chromatic index.

\begin{rmk}\label{rem2} 
Unlike the usual edge coloring problem, there are some graphs $G$ with an improper twin $k$-edge coloring $s$  such that the restriction of $s$ to $E(G)- e$ for some edge $e$ is not necessarily an improper twin $k$-edge
coloring of $G-e$. \end{rmk}

Indeed, a triangle admits an improper twin 3-edge coloring with colors 0, 1 and 2 on its three edges, however, removing the edge of color 2 leaves a graph with an edge coloring which does not induce a proper vertex coloring.  Note however that removing an edge colored with 0 obviously leaves an improper twin $k$-edge coloring in the remaining graph.

\begin{rmk}\label{rem3}
Unlike the usual vertex/edge coloring problems, an improper twin $k$-edge coloring is not necessarily an improper twin $t$-edge coloring for  $t\geq k$.  
\end{rmk}

This is due to the $\mod t$ or $\mod k$ operations respectively used in the definition of twin $t$- and twin $k$-edge coloring.
Consider indeed a star with a center vertex $a$ and four pending vertices $b,c,d,e$. The edge coloring assigning color 1 to all the four edges is an improper twin 2-edge coloring but not an improper twin 3-edge coloring. 

\subsection{Our contribution}

In this paper we investigate the improper twin edge chromatic index of graphs and we consider, to the best of our knowledge, the first time, the complexity of a neighbor-distinguishing edge coloring problem. 
The following summarizes the main contribution of \cite{JCMCC11} and \cite{JCMCC12}:

\begin{thm}(\cite{JCMCC11, JCMCC12}) \label{thm:mainfromJCMCC}
Let $G$ be a connected graph of order at least 3. Then $\chi'_{it}(G) =\chi(G) +1$ if and only if $\chi(G) \equiv 2 \pmod 4$ and every proper $\chi(G)$-vertex coloring of $G$ results in color classes of odd size; and $\chi'_{it}(G) =\chi(G)$ in all other cases.
\end{thm}

It should be noted that this result is obtained as a combination of several results from  \cite{JCMCC11} and \cite{JCMCC12}. In Section \ref{sec:general}, we provide a shorter and more compact proof of Theorem \ref{thm:mainfromJCMCC}.  
Actually we obtain a more precise result	
(see Theorem~\ref{thm:final}) describing a constructive way to build in linear time an improper twin edge coloring using $k$ or $k+1$ colors, whenever a $k$-vertex coloring is provided.


In~\cite{JCMCC12} it is shown that for non-bipartite connected graphs $G$ and odd $t\geq \chi(G)$, any $t$-vertex coloring  can be induced by an improper $t$-edge coloring. We give an alternative proof of this result (Proposition~\ref{oddc}) and give 
a necessary and sufficient condition for a $t$-vertex coloring with even $t$ being induced by an improper twin $t$-edge coloring in a non-bipartite connected graph (Proposition~\ref{equal}).

Moreover, these two results, together with Theorem~\ref{bip} 
 allow us to establish the following \emph{monotonicity} property of the improper twin chromatic index for graphs: if a graph $G$ admits an improper twin $k$-edge coloring for some positive integer $k$, then it admits an improper twin $t$-edge coloring for all $t\geq k$ (Theorem~\ref{thm:at least k}). Note that following Remark~\ref{rem3}, unlike the classical coloring problems, this property is not straightforward for the improper twin edge coloring problem.


The characterization of graphs having $\chi_{it}'(G)=\chi(G)+1$ emphasizes  the closely related problem  of deciding whether all color classes are of  odd size in all optimal colorings. In Section \ref{sec:complexity}, we use this new problem to show that, interestingly enough, just like in the case of usual edge chromatic index which can only take two possible values, it is NP-complete to decide whether $\chi'_{it}(G)= \chi(G)$ even when the graph $G$ is restricted to have $\chi(G)$ equal to its maximum clique size and in the bounded degree case.   

Finally, in Section \ref{sec: poly}, we give first examples of graph classes for which this question can be decided in polynomial time. These classes include bipartite graphs, split graphs, co-chordal graphs and co-comparability graphs with bounded chromatic number.

\section{Bounds and Monotonicity for the Improper Twin Edge Chromatic Index}\label{sec:general}
The contribution of this section is two-fold. First, we provide a short and compact proof of the bounds $\chi(G) \leq \chi'_{it}(G)\leq \chi(G)+1$ together with a description of the graphs having $\chi'_{it}(G)= \chi(G)+1$. Second, we prove the monotonicity of the improper twin chromatic index.

We first settle both the monotonicity and the bounds for bipartite graphs. More precisely, we show that for every nice bipartite graph $G$ and every positive integer $t\geq 3$, $G$ admits an improper twin $t$-edge coloring.

Given a connected graph $G=(V,E)$, a $t$-vertex coloring  $c$ and an improper twin edge coloring $s$ using colors in $\mathbb{Z}_t$, we will say that $s$ {\em almost induces $c$} if there is a vertex $v\in V$ such that $\forall u\in V\setminus\{v\}, c(u)=c_s(u)$. In this case $v$ is called the {\em defective vertex}. Note that 
	if $s$ induces $c$ then it almost induces it and any vertex can be chosen as defective. Moreover if $s$ is almost inducing $c$, either it induces it or there is a unique possible defective vertex. The following lemma plays a crucial role in our approach. 
\begin{lem}\label{lem:tree}
	Given connected graph $G=(V(G),E(G))$, a vertex $t$-coloring $c$ 
of $G$ ($t\geq 1$), a vertex $v\in V$ and a spanning tree $T$, there is an improper twin edge coloring $s$ that almost induces $c$ with $v$ as defective vertex. Moreover $\forall  e\in E(G)\setminus E(T), s(e)=0$.
\end{lem}
	
	\begin{proof}
We direct the edges of $T$ from the root to the leaves using $v$ as the root. We set $\forall  e\in E(G)\setminus E(T), s(e)=0 $. Then, starting from the deepest level of $T$ and going up to the root, color the edges $e$ of $T$ with colors $s(e)$ in $\{0, \ldots ,t-1\}$ to ensure that for every vertex $u\neq v$, the color of $u$ corresponds to the sum in $\mathbb{Z}_t$ of the colors of edges of $T$ incident to $u$. Since $s(e)=0$ for the edges 
$e\in E(G)\setminus E(T)$, it means that $\forall u\neq v, c(u)=c_s(u)$. 
	\end{proof}

In the configuration described in Lemma~\ref{lem:tree}, we will say that $s$ is {\em associated with $T$ and $v$}. Note that this notion is not unique: several different improper twin $t$-edge colorings can be associated with the same tree $T$ and $v$.  	
	
	\begin{thm}\label{bip} 
If $G=(X\cup Y,E)$ is a nice bipartite graph with vertex parts $X$ and $Y$, then  the following statements hold:
\begin{enumerate}
\item[i)] \cite{JCMCC11} $\chi'_{it}(G)=2$ if and only if 
$G$ does not have any connected component with both parts of the bipartition of odd size.
\item[ii)] 
For any integer $t\geq 3$ the graph $G$ admits an improper twin $t$-edge coloring.
\end{enumerate}
\end{thm}

\begin{proof}
Without loss of generality, suppose that $G$ is connected with at least 3 vertices. 
For a vertex $v\in X\cup Y$ consider a spanning tree $T$ of $G$ and direct the edges from the root of the leaves using $v$ as the root. Set $c(x)=1$ for all vertices $x$ in the same part of the bipartition as $v$  and $c(x)=0$ for all vertices in the other part. Using Lemma~\ref{lem:tree}, we consider an improper twin $t$-edge coloring $s$ of $G$ associated with $T$ and $v$.

$i)$ Consider $t=2$. Without loss of generality we assume 
 that $|X|$ is even and then, we chose $v\in X$.  We then have:

\begin{equation}\label{eq:bip}
	0\equiv \sum_{w\in X}c(w) \equiv \sum_{u\in Y}c(u) \pmod 2,
\end{equation}


 where the first summation is the sum of an even number of 1's and  the second summation is the sum of 0's. 
 
	On the other hand, by construction of $c$ we have:
 
 \begin{equation}\label{eq:bip2}
	\sum_{u\in Y}c(u)\equiv \sum_{e\in E(T)}s(e) \equiv  \left(\sum_{u\in X\setminus\{v\}}c(u)+ \sum_{e\in E_v}s(e)\right) \pmod 2.
\end{equation}

Using Relations~\ref{eq:bip} and~\ref{eq:bip2} we deduce $c(v) \equiv \sum_{e\in E_v}s(e) \pmod 2=c_s(v)$. So, $c=c_s$ and $s$ is an improper twin 2-edge coloring of $G$.

To prove the other direction, we proceed by contradiction: assume that $|X|$ and $|Y|$ are both odd and we have an improper twin 2-edge coloring of $G$. Since $G$ is connected, $(X,Y)$ is the unique bipartition of $G$ which should be also induced by the improper twin 2-edge coloring. By double-counting the color of all edges as previously, we obtain a contradiction to the fact that both $|X|$ and $|Y|$ are odd.\\

$ii)$  We assume now that $v$ is of degree at least~2 in $G$ ($G$ is nice) and consider $T$ such that $v$ is of degree at least~2 in $T$.

Let $a_1, a_2, \ldots a_r$
  be the colors (among $0,1, \ldots ,t-1$) of the edges of $T$   incident to $v$ ($r\geq 2$).  If $a_1+ a_2+ \ldots + a_r \not\equiv 0 \pmod t $, then $s$  is an improper twin $t$-edge coloring of $T$ that almost induces $c$.   Now, if $a_1+ a_2+ \ldots + a_r \equiv 0 \pmod t$, then we consider two cases. If $t\neq 4$, then change the color of edge $e_1=vx_1$ from $a_1$ to $a_1+2 \mod t$ and the color of edge $e_2=vx_2$ from $a_2$ to $a_2+2 \mod t$ for some $x_1,x_2 \in N_T(v)$ to obtain an improper twin $t$-edge coloring $s'$.  $c_{s'}(x_1)=c_{s'}(x_2)=2$ and $c_{s'}(v)=4 \mod t$ (note that $4 \mod t \notin\{0,2\}$). 
  If $t=4$, then change the color of edge $e_1=vx_1$ from $a_1$ to $a_1+3 \mod 4$ and change the color of edge $e_2=vx_2$ from $a_2$ to $a_2+3 \mod 4$ to obtain an improper twin $t$-edge coloring $s'$. We have  $c_{s'}(x_1)=c_{s'}(x_2)=3$ and $c_{s'}(v)=2$. 
\end{proof}

Note that in the case $ii)$, the induced $t$-vertex coloring has at most 3 vertices with a color greater than~1. 
The following is a key lemma for the establishment of the monotonicity property and the bounds for non-bipartite case in a constructive way.

\begin{lem}\label{odd-cycle}
Let $G=(V,E)$ be a connected non-bipartite graph. Let $c$ be a $k$-vertex coloring of $G$ almost induced by an improper twin $k$-edge coloring $s$ with defective vertex $v$. For any $x\in \mathbb{Z}_k$ such that   $\forall 
u\in N_G(v), c(u)\neq c_s(v)+2x \mod k$, the $k$-vertex coloring obtained by replacing $c(v)$ with $c_s(v)+2x \mod k$ is induced by an improper twin $k$-edge coloring.
  \end{lem}

\begin{proof}
Since $\chi(G)>2$, the graph $G$ contains an odd cycle, say $C$. Since $G$ is connected, there exists a path between $v$ and $C$, say $P$. We construct an odd closed walk $W$ passing by $v$ as follows: traverse $P$ starting at $v$ and go to $C$ and then traverse $C$ and come back to $v$ via $P$. Note that each edge of $P$ is traversed twice.

Given $x\in \mathbb{Z}_k$, add alternately $+x, -x \mod k$ to $s(e)$,  for each $e$ in $W$ starting from $v$ and consider the vertex coloring induced by this new edge coloring $s'$. All vertices $u\in W, u\neq v$  have the same color since the summation of the colors of edges incident to $u$ is the same ($s(u)=s'(u)$). Vertices outside $W$ are not affected at all and the color $c(v)$ is changed into $c_s(v)+2x \mod k = c_{s'}(v)$. This new vertex coloring is a  proper vertex coloring since the color $c(v)+2x \mod k$ is not used in $N_G(v)$ and it is induced by $s'$. \end{proof}

From now on, we deal with the 
non-bipartite case. We first settle the bounds and the monotonicity for all improper twin edge colorings with odd number of colors. 
Lemmas~\ref{lem:tree} and~\ref{odd-cycle} give an alternative proof of the following result.

\begin{prop}\label{oddc}\cite{JCMCC12}  If $G$ is a nice graph without bipartite component and $t$ is an odd positive integer such that $t\geq \chi(G) \geq 3$, then 
any $t$-vertex coloring  of $G$ with  colors in $\{0,1, \ldots,t-1\}$ can be induced by an improper twin $t$-edge coloring of $G$.
\end{prop}

\begin{proof}
We assume without loss of generality that $G$ is connected since otherwise the same proof applies to every connected component (note that if the hypotheses hold for $G$ they hold for each component). Let $f$ be a $t$-vertex coloring of $G$ with  colors in $\{0,1, \ldots,t-1\}$ with $t\geq\chi(G)$. We show that there is an improper twin $t$-edge coloring $s$ of $G$ satisfying $c_s(v_i)=f(v_i)$ for all vertices $v_i \in V(G), i=1,\ldots, n$.


Consider a spanning tree $T$ of $G$ and direct the edges from the root of the leaves using $v_1$ as the root. 
Using Lemma~\ref{lem:tree}, we consider an improper $t$-edge coloring $s$ of $G$ associated with $T$ and $v$. It satisfies 
  $\forall i\geq 2, c_s(v_i)=f(v_i)$. 
  Let $a_1, a_2, \ldots ,a_r$ be the colors (among $\{0, \ldots ,t-1\}$) of the edges in $E(T)$ incident to $v_1$.  We have $c_s(v_1)= a_1+ a_2+ \ldots +a_r \mod t$. If $c_s(v_1) = f(v_1)$, 
then $f$ is induced by the improper twin $t$-edge coloring  $s$. Otherwise, since  $k$ is odd ($\gcd(k,2)=1$), the equation $c_s(v_1)+2x=f(v_1) \mod k$ has a solution $x\in \mathbb{Z}_k$.  Lemma~\ref{odd-cycle} ensures that we can modify $s$ into $s'$ such that $f=c_{s'}$ (replacing $f(v_1)$ by $c_s(v_1)+2x$ does not change the vertex coloring $f$). It completes the proof.
 \end{proof}

\begin{rmk}
	Note that the non-bipartite condition in Proposition~\ref{oddc} is required. In particular we cannot improve Theorem~\ref{bip} ii) in the sense that, for bipartite components, only some  $t$-vertex colorings are induced by an improper twin $t$-edge twin coloring. 
\end{rmk}
Indeed, by Theorem~\ref{bip} i) a 2-vertex coloring of the star graph $K_{1,3}$ cannot be induced by an improper twin 2-edge coloring. Moreover, for any even $p\geq 4$, the $p$-vertex coloring of the star graph $K_{1,p-1}$ assigning color 0 to the center and different colors from 1 to $p-1$ for the other vertices is not induced by an improper twin $p$-edge coloring since $\frac{p(p-1)}{2}\neq 0 \mod p$. For any odd $p\geq 3$, the $p$-vertex coloring of the star graph $K_{1,p-1}$ assigning color 1 to the center and different colors in $\{0, 2, \ldots ,p-1\}$  for the other vertices is not induced by an improper twin $p$-edge coloring since $\left(\frac{p(p-1)}{2}-1\right)\neq 1 \mod p$.



The following proposition will enable us to obtain the monotonicity of the improper twin edge coloring for any number of colors (Theorem \ref{thm:at least k}) and the bounds on $\chi_{it}'(G)$ for graphs with even chromatic number at least 4.

\begin{prop}\label{equal}
Let $G$ be a nice non-bipartite connected graph and $t$ an even number such that $t\geq \chi(G)\geq 3$. Then, a $t$-vertex coloring $f$ with colors in $\{0, \ldots, t-1\}$ is induced by an improper twin $t$-edge coloring 
if and only if 
$\sum_{v\in V(G)}f(v)$ is even.
\end{prop}

\begin{proof}
$(\Leftarrow)$  
Denote $n=|V(G)|$ and $V=\{v_1, \ldots, v_n\}$ and let $f$ be a $t$-vertex coloring of $G$ with colors in $\{0, \ldots, t-1\}$ such that $t$ and $\sum_{i}^{n}f(v_i)$ are even. 

As in the proof of Proposition~\ref{oddc} we consider a spanning tree $T$ of $G$ and direct the edges from the root of the leaves using $v_1$ as the root. Using Lemma~\ref{lem:tree}, we consider an improper $t$-edge coloring $s$ of $G$ associated with $T$ and $v_1$. It satisfies 
  $\forall i\geq 2, c_s(v_i)=f(v_i)$.

 If $c_s(v_1)=f(v_1)$, then the proof is complete. Else we claim that $f(v_1)-c_s(v_1)$ is even. Indeed, by definition of $s$ and denoting by $E'\subset E$ the set of edges in $E(T)$ not adjacent to $v_1$ we have:
 	
 	\begin{equation}\label{sum-vertices}
 		\sum_{i=2}^n f(v_i)=c_s(v_1)+2\sum_{e\in E'}s(e)
 	\end{equation}

 	Since $\sum_{v\in V(G)}f(v)$ is even, this implies $c_s(v_1)+f(v_1)$ is even and consequently $f(v_1)-c_s(v_1)$ is even as well. We set 
 	$$x= \frac{f(v_1)-c_s(v_1)}{2} \mod t,$$ so $f(v_1)=c_s(v_1)+2x \mod t$ and use Lemma~\ref{odd-cycle} to modify $s$ in $s'$ such that $f=c_{s'}$. It completes the proof of the first implication.
 	
$(\Rightarrow)$  Conversely, suppose that $\chi'_{it}(G)=\chi(G)$. Consider an improper twin edge coloring $s$ using the colors $\{0, 1, \ldots, \chi(G)-1\}$. This induces a vertex coloring $f$. We have $\sum_{v\in V(G)}f(v)\equiv 2(\sum_{e\in E(G)}s(e))  \pmod k.$ Now, since $k$ is even,
the proof is complete.
\end{proof}

Now, we are ready to show the monotonicity property of the improper twin edge chromatic index for the general case.

\begin{thm}\label{thm:at least k}
Let $G$ be a nice graph. If $G$ admits an improper twin $k$-edge coloring ($k\geq 2$), then $G$ admits an improper twin $t$-edge coloring for all $t\geq k$. 
\end{thm}
\begin{proof}
We assume without loss of generality that $G$ is connected since otherwise the same proof applies to every connected component. The assertion vacuously holds for $t=k$ for all graphs. If $G$ is bipartite, the assertion holds by Theorem~\ref{bip}. Assume now $G$ is not bipartite and $t>k\geq 3$.
If $t$ is odd, then the assertion holds by Proposition~\ref{oddc}.  
Assume $t$ is even. Obviously, $\chi(G)\leq k$. Consider a proper vertex coloring of $G$, say $f:V(G) \rightarrow \{0, \ldots, k-1\}$. If the sum $\sum_{v\in V(G)}f(v)$ is even then Proposition~\ref{equal} allows to conclude. 
 Otherwise, consider the maximum color used in $f$, say $\ell$ and change the color of exactly one vertex with color $\ell$ to $\ell+1$ and call $f'$ the resultant proper coloring. 
  Clearly, $\sum_{v\in V(G)}f'(v)$ is even and apply Proposition~\ref{equal} to conclude the proof.
  \end{proof}


Let us finally note that  Lemma~\ref{odd-cycle} and its consequence Proposition~\ref{equal} give a comprehensive and direct proof of Theorem~\ref{thm:mainfromJCMCC} in the case where $\chi(G)$ is even of size at least~4, as stated below. The other cases are immediate consequences of Theorem~\ref{bip} and Proposition~\ref{oddc} for which we gave alternative simple proves as well. Moreover, as stated in Theorem~\ref{thm:final} the arguments can be turned into a linear-time algorithm.

\begin{proof}[Alternative Proof of Theorem~\ref{thm:mainfromJCMCC}] \mbox{}\\
If $G$ is bipartite, the statement holds by Theorem~\ref{bip}. For a non-bipartite connected graphs $G$ of order at least~3, if $\chi(G)$ is odd then the statement holds by 
Proposition~\ref{oddc}. 
 To complete the proof of Theorem~\ref{thm:mainfromJCMCC}, it remains to show that the statement holds for $\chi(G)$ even. Let us show this under two cases.

\textbf{Case 1: $\chi(G)\equiv 0 \pmod 4$}. Then, for some positive integer $r$, there are  $2r$ colors 
with an even label and the same holds for 
 colors of odd labels. Consider a  $\chi(G)$-vertex coloring of $G$, say $f$. If all color classes have odd size, then $\sum_{v\in G}f(v)$ is even and by Proposition~\ref{equal} ($G$ is not bipartite) we are done. Similarly, if all classes have even size we are done as well. Thus  assume that
we have $\chi(G)$ color classes $C_0, \ldots, C_{\chi(G)-1}$ in which, without loss of generality, $|C_0|$ and $|C_1|$ have different parity.
Now, color $C_i, i=2, \ldots, \chi(G)-1$, arbitrarily using the colors $2, 3, \ldots, \chi(G)-1$.  We have two possibilities
for coloring $C_0$ and $C_1$ by two colors $0$ and $1$, and obviously for one of them, the total sum of colors of vertices are even and Proposition~\ref{equal} completes the proof.


\textbf{Case 2: $\chi(G)\equiv 2 \pmod 4$}. Then there is an odd number of odd colors and the same odd number of even colors in a $\chi(G)$-vertex coloring of $G$. We will show that if there is a color class of even size in some $\chi(G)$-vertex coloring of $G$ then $\chi'_{it}(G)=\chi(G)$, and otherwise, that is if all color classes in all $\chi(G)$-vertex colorings of $G$ are of odd size, then $\chi'_{it}(G)=\chi(G)+1$.\\
Consider a  $\chi(G)$-vertex coloring of $G$, say $f$. If all color classes have even size, then obviously $\sum_{v\in G}f(v)$
is even and by Proposition~\ref{equal} we are done. If $f$ has two color classes, say without loss of generality $C_1$ and $C_2$ whose cardinalities are of different parity, then Proposition~\ref{equal} allows to conclude exactly as in part (i).


Now, assume that all color classes in all $\chi(G)$-vertex coloring of $G$ are of odd sizes. Then $\sum_{v\in V(G)} f(v)$ is odd for all $\chi(G)$-vertex colorings $f$ of $G$ and by Proposition~\ref{equal}, $\chi'_{it}(G)>\chi(G)$. Now, $\chi(G)+1$ is an odd number
  and by the method used in the proof of Proposition~\ref{oddc}, one can obtain an improper twin edge coloring of $G$ with $\chi(G)+1$ colors, implying that $\chi'_{it}(G)=\chi(G)+1$.
\end{proof}

Finally, the following theorem shows that the  arguments can be turned into a linear-time algorithm which constructs an improper twin $k$ or $(k+1)$-edge coloring whenever a $k$-vertex coloring is given.


\begin{thm}\label{thm:final}
There is a $O(|V|+|E|)$ algorithm that computes, for 
a nice graph $G=(V,E)$ and a $k$-vertex coloring  of $G$, $k\leq n$:\\
i) if $k\equiv 2 \pmod 4$ and there is at least one connected component with an odd number of vertices of each color, an improper twin $(k+1)$-edge coloring;\\
ii) else an improper twin $k$-edge coloring.
\end{thm}

\begin{proof}
Given the graph $G=(V,E)$ the connected components can be determined in $O(|V|+|E|)$; we can assume $G$ is connected and apply, for a non connected graph, the same algorithm on each connected component. There is a $O(|V|+|E|)$ algorithm deciding whether $G$ is bipartite and, in this case, computing its bipartition.
 \noindent Note first that operations in $\mathbb{Z}_k$ can be performed in $O(k)$ time and $k\leq |V|$. In particular, if $k$ is odd, determining $x$ such that $2x+y=z$ requires for instance the Euclidean algorithm~\cite{discrete-notes}.\\
  
 \noindent
 Lemma~\ref{lem:tree} (including the determination of the spanning tree $T$) gives a $O(|V|+|E|)$ algorithm to compute, for a $t$-vertex coloring $c$ of $G$ and $v\in V$ an improper $t$-edge coloring that almost induces $c$ with $v$ as defective vertex.
 Indeed, it requires ( 1) computing a spanning tree $T$, (2) assigning the color 0 to edges outside the trees and (3) computing the colors of the tree-edges proceeding from the leaves to the root. All these operations can be respectively performed in $O(|V|+|E|), O(|E|)$ and $O(|V|+|E|)$.\\

 \noindent
 Suppose that $G$ is bipartite.\\
 If $k=2$ and one part of $G$ is of even size or if $k\geq 3$, we use the method described  in the proof of Theorem~\ref{bip} to compute an improper twin $k$-edge coloring inducing $c$. It requires an application of Lemma~\ref{lem:tree} with the right choice of the root $v$ (either of degree 2 if $k\geq 3$ or in an even part of the bipartition if $k=2$). This choice requires $O(|V|)$ operations.
 If $k=2$ and both parts of the bipartition are odd, then we apply the previous method replacing $k$ by 3.\\
 
 \noindent
Let us now suppose 	$G=(V,E)$ is not bipartite (of course $k\geq 3$).\\
\noindent  
 Lemma~\ref{odd-cycle} gives a $O(|V|+|E|)$ strategy to modify a given $k$-edge coloring accordingly.  It requires determining an odd-cycle $C$, a path $P$ from $v$ to $C$ and performing modifications along the related closed walk $W$. These operations require respectively $O(|V|+|E|)$,  $O(|V|+|E|)$ and $O(|E|)$ operations.\\

\noindent
If $k$ is odd, then we revisit the proof of Proposition~\ref{oddc} to compute an  improper twin $k$-edge coloring. It requires one application of Lemma~\ref{lem:tree} and eventually one equation in $\mathbb{Z}_k$ and one application of Lemma~\ref{odd-cycle}.

\noindent
If  $k\equiv 0 \pmod 4$ or $k\equiv 2 \pmod 4$, $k\geq 4$, then   we use similar arguments as in the proof of Theorem~\ref{thm:mainfromJCMCC}, replacing $\chi(G)$ with $k$. If $k\equiv 0 \pmod 4$ or if $k\equiv 2 \pmod 4$ and at least one color class of even size, then determining an improper twin $k$-edge coloring requires, using the method in Proposition~\ref{equal}, an application of Lemma~\ref{lem:tree}  and, eventually, one application of Lemma~\ref{odd-cycle}. Finally, if $k\equiv 2 \pmod 4$ and all color classes are of odd size, then we use the case where $k$ is odd considering the $k$-vertex coloring as a $(k+1)$-vertex coloring. 
\end{proof}

Since the chromatic number of planar graphs is bounded by 4 \cite{4CT, planar}, we trivially have the following:

\begin{cor}
Let $G$ be a nice connected planar graph, 
 then  $\chi'_{it}(G)\leq 4 $.
\end{cor}

\section{Hardness of the minimum improper twin edge coloring}\label{sec:complexity}

This section is devoted to some complexity results for the problem of deciding, given a graph $G$,  whether  $\chi'_{it}(G)=\chi(G)$. We restrict ourselves to the class of instances for which a clique of size $\chi(G)$ as well as a $\chi(G)$-vertex coloring are given, which allows as to guarantee the polynomial recognition of an instance, a usual assumption in the area of NP-hardness results. More formally we consider the following problem:\\

 {\sc Improper Twin  $k$-Edge Coloring} \\
Input: A graph $G$ with a clique of size $k$ and a $k$-vertex coloring of $G$.\\
Question: Is $\chi'_{it}(G)=\chi(G)$?\\

First, we notice that from the definition of {\sc Improper Twin $k$-Edge Coloring}, an instance $G$ has $\chi(G)=k$ and this is equal to the size of a largest clique of $G$. This problem is clearly in NP since, given a ``yes" answer with an improper twin  $\chi(G)$-edge coloring of $G$, it takes linear time to check whether the $\chi(G)$-vertex coloring induced by this improper edge coloring is proper. In addition, we will show in the sequel that the following problem, called {\sc All Odd Optimal $k$-Vertex Colorings} is coNP-complete. This implies by Theorem \ref{thm:final} that {\sc Improper Twin $k$-Edge Coloring} is NP-complete for $k\equiv 2\pmod 4$; indeed these are complementary problems in the sense that the answer is ``yes" for one of them if and only if the answer is ``no" for the other.\\

 {\sc All Odd Optimal $k$-Vertex Colorings}\\
Input: A graph $G$ with a clique of size $k$ and a $k$-vertex coloring of $G$.\\
Question: Are all color classes of odd size in all $\chi(G)$-vertex colorings of $G$?\\

We  need the following result for our reduction:
\begin{lem}\label{lem:gadget}
Let $\mathcal J$ be the graph in Figure \ref{fig:gadget} (a) and consider a 3-vertex coloring of $\mathcal J$ with colors F, T and R where vertices $x,y,z$ and $X$ are colored either $T$ or $F$. Then the following hold:
\begin{enumerate}
\item[(i)] In any such 3-vertex coloring of $\mathcal J$, where $X$ is colored $T$, at least one among $x,y,z$ should be colored $T$ and among $a_1, a_2, \ldots, a_5$, there are $1T, 2F$ and $2R$ colors.
\item[(ii)] Similarly, in any such 3-vertex coloring of $\mathcal J$, where $X$ is colored $F$, at least one among $x,y,z$ should be colored $F$ and among $a_1, a_2, \ldots, a_5$, there are $1F, 2T$ and $2R$ colors.
\end{enumerate}

\end{lem}
\begin{proof}
We only prove statement (i) since (ii) is obtained similarly. Let $f$ be a $3$-vertex coloring of $\mathcal J$, where $f(X)=T$ and each one of  $f(x),f(y),f(z)$ is in $\{T,F\}$. Then at least one of $x,y,z$ has color $T$, since otherwise (that is if all of $x,y,z$ are colored $F$) we necessarily have $f(a_3)=F$, implying $c(a_4)=R$. Now, $a_5$ has all three colors in its neighbourhood, hence it is impossible to color $a_5$ with one of the three colors $T, F$ and $R$. Besides, $f(X)=T$ implies that we have either $c(a_5)=F, c(a_4)=R$ or $c(a_5)=R, c(a_4)=F$. In either case, among $a_1, a_2, a_3$ there are exactly one $T$, one $F$ and one $R$. Therefore, there are in total exactly $1T, 2F$ and $2R$ colors among $a_1,a_2, \ldots, a_5$.
\end{proof}

\begin{figure}[htbp]
\begin{center}
\includegraphics[width=11cm]{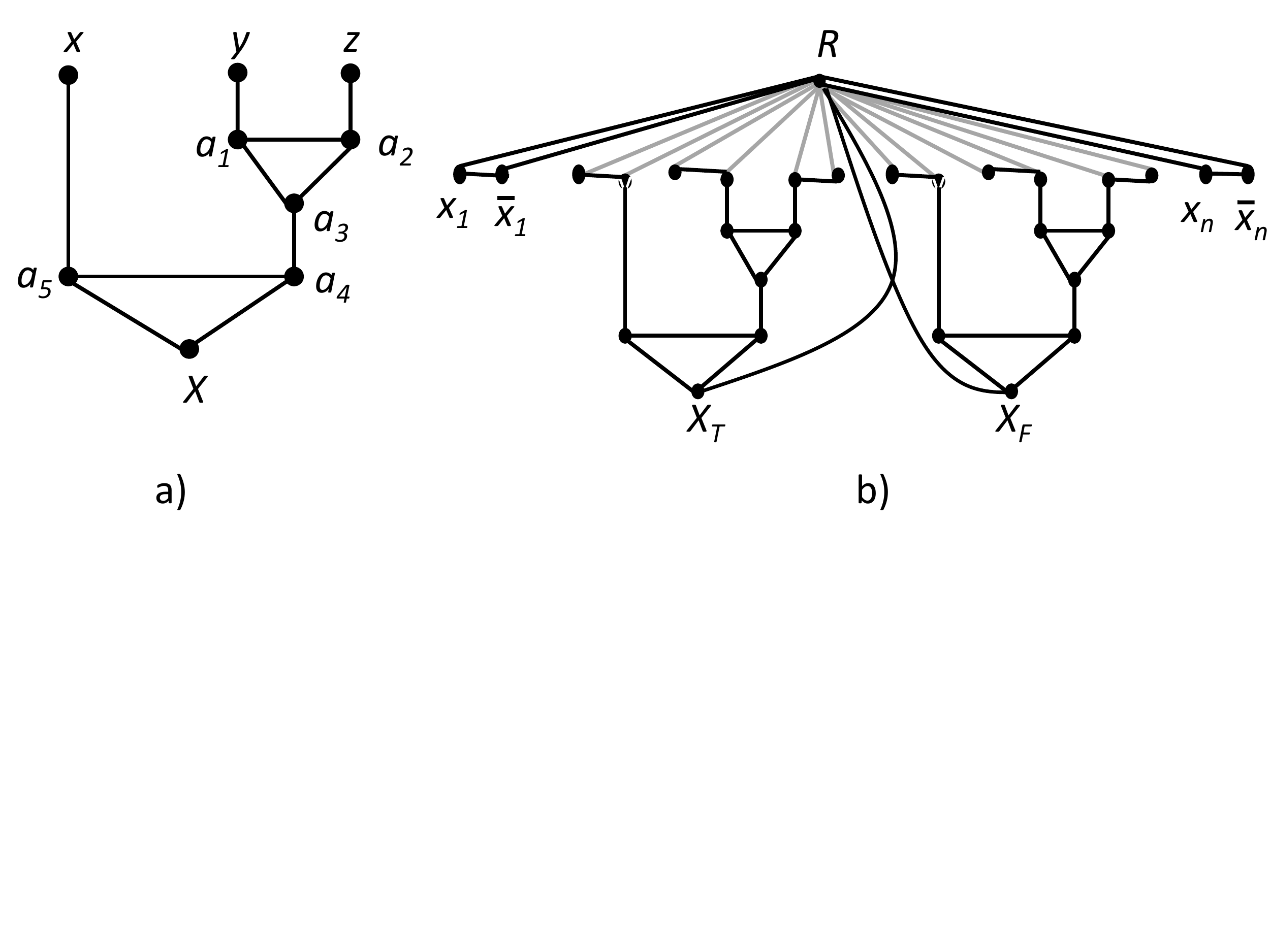}
\vspace{-3,5cm}
\caption{The Clause Gadget $\mathcal J$ and the instance $G$.}\label{fig:gadget}
\end{center}
\end{figure}

\begin{thm}\label{coNP}
 For all $k\geq 3$, {\sc All Odd Optimal $k$-Vertex Colorings} is coNP-complete.
\end{thm}
\begin{proof}
First,  {\sc All Odd Optimal $k$-Vertex Colorings} is in coNP, since one can verify the correctness of a ``no" certificate in linear time; in other words, given a $k$-vertex coloring with at least one color class of even size, one can check whether it is a vertex coloring using at most $k$ colors and whether it contains a color class with even number of vertices. In the sequel, we show that 3-SAT can be polynomially reduced to the complement of {\sc All Odd Optimal $k$-Vertex Colorings} with $k=\chi(G)=3$.

Let $I$ be an instance of 3-SAT with $n$ variables and $m$ clauses. We assume without loss of generality that $n$ is even (if not, add a variable that appears in no clause). Consider a truth assignment (for instance where all variables are true) and let $m_T$ and $m_F$ be respectively the numbers of true and false clauses. Denote also by $M_T$ and $M_F$ the corresponding sets of clauses, respectively.  Assume without loss of generality that both  $m_T$ and $m_F$ are even (otherwise just double each clause in $I$). We construct an instance $G$ of the complement of  {\sc All Odd Optimal $k$-Vertex Colorings} with $k=\chi(G)=3$ as follows (see Figure \ref{fig:gadget} (b)). For every variable, consider an edge between vertices $x_i$ and $\bar x_i$ representing the variable and its negation. Introduce two vertices $X_T$ and $X_F$ representing respectively clauses in $M_T$ and clauses in $M_F$. For every clause, consider a clause gadget $\mathcal J$ as in Figure \ref{fig:gadget} (a), whose $x,y,z$ vertices are identified with the variable vertices representing the tree literals (variables or their negations) in the related clause, and whose $X$ vertex is identified with $X_T$ if it is a clause in $M_T$ and with $X_F$ otherwise. Finally, add a vertex $R$ adjacent to all variable vertices and to $X_T$ and $X_F$. Clearly, this instance $G$ can be constructed in polynomial time in the size of the 3-SAT instance $I$. Moreover, it is $3$-vertex colorable; color vertex $R$ by color $R$, $X_T$ by $T$ and $X_F$ by $F$, color also all true literals by $T$ and all false literals by $F$. Now, one can check that this 3-vertex coloring can be extended to the clause gadgets corresponding to clauses in $M_T$ for all possible cases with 1, 2 or 3 true literals per clause. Similarly, it can be extended to the clause gadgets corresponding to clauses in $M_F$ and thus having all variable vertices colored $F$. Indeed, the only 3-vertex coloring of the three literals and $X_T$ (or $X_F$) with colors $T$ and $F$ which can not be extended to a 3-vertex coloring of the clause gadget is when the three literals are all assigned the color $F$ and $X_T$ is assigned $T$, or when the three literals are all assigned the color $T$ and $X_F$ assigned $F$; clearly, this never occurs in the 3-vertex coloring that we defined.

\begin{claim}
Let $I$ be a 3-SAT instance and $G$ a graph constructed as described above. Then, $I$ is satisfiable if and only if $G$ has a 3-vertex coloring with $X_T,X_F$ and variable vertices colored with $T$ or $F$, where at least one color class has even size.
\end{claim}
\begin{proof}
Consider a truth assignment satisfying all clauses of $I$. Then $G$ can be 3-colored with $T, F$ and $R$ as follows. Color $X_T, X_F$  and all vertices representing true literals by $T$. Color the vertex $R$ by color $R$. Now, one can check that for all possible cases with 1, 2 or 3 true literals per clause,  this coloring can be extended to a 3-vertex coloring of $G$. Moreover, by Lemma~\ref{lem:gadget} and by the fact that in any such 3-vertex coloring of $G$ there are exactly $n$ variable vertices assigned color $T$, in this 3-vertex coloring there are in total $n+m+2$ vertices colored by $T$, which is an even number (remind that $n, m_T$ and $m_F$ are even and $m=m_T+m_F$).

Now, consider a 3-vertex coloring $f$ of $G$ with $X_T, X_F$ and all variable vertices are colored with $T$ or $F$, where there is at least one color class of even size.   
Note that in any such 3-vertex coloring of $G$, there are $n$ variable vertices of color $T$ and $n$ variable vertices of color $F$. If $X_T$ and $X_F$ are colored with different colors, say by $T$ and $F$, respectively, then the numbers of occurrences of $T, F$ and $R$ are, respectively $n+m_T+2m_F+1, n+m_F+2m_T+1$ and $2m+1$, where $m=m_T+m_F$. The other case, where $f(X_T)=F$ and $f(X_F)=T$ yields the following numbers of occurrences for $T, F$ and $R$, respectively: $n+m_F+2m_T+1, n+m_T+2m_F+1$ and $2m+1$. It can be seen that all these numbers are odd, a contradiction.  Therefore $f(X_T)=f(X_F)$, and, using Lemma~\ref{lem:gadget}, a truth assignment satisfying all clauses can be obtained by giving the value true to the literals having the same color as $X_T$ and $X_F$. 
\end{proof}

Since 3-SAT is NP-complete~\cite{gj}, deciding whether a given graph with a clique of size 3 and a 3-vertex coloring admits an optimal vertex coloring (with 3 colors) with at least one color class of even size is NP-complete. It follows that  {\sc All Odd Optimal $k$-Vertex Colorings} is coNP-complete for $\chi(G)=k=3$. The coNP-completeness of  {\sc All Odd Optimal $k$-Vertex Colorings} for $\chi(G)=k\geq 3$ is obtained simply by adding a clique of size $k-3$ completely linked to the graph $G$ in the above reduction; clearly, the $k-3$ vertices of this clique should receive colors other than $T, F$ and $R$ and thus will not affect the above facts.
\end{proof}

Let $G$ be an instance of {\sc Improper Twin $k$-Edge Coloring} such that $k \equiv 2 \pmod 4$ (which can be checked in constant time). Now, it follows from Theorem \ref{thm:final} that the answer to  {\sc All Odd Optimal $k$-Vertex Colorings} for $G$ is ``yes" if and only if the answer to {\sc Improper Twin $k$-Edge Coloring} is ``no". In addition, we have already noticed that {\sc Improper Twin $k$-Edge Coloring} is in NP. Therefore {\sc Improper Twin $k$-Edge Coloring} and {\sc All Odd Optimal $k$-Vertex Colorings} are complementary problems and Theorem~\ref{coNP} and Theorem~\ref{thm:final} imply the following:

\begin{cor}
For all $k\geq 6$, $k \equiv 2 \pmod 4$, {\sc Improper Twin $k$-Edge Coloring} is NP-complete. It is polynomial for all other values of $k$.
\end{cor}

The proof of Theorem~\ref{coNP} opens some natural questions about restricted hard and polynomial cases. We give first results in this direction that underline some open questions for future research. The first natural particular case is the class of graphs of bounded degree. The next proposition shows that, for any fixed $k\geq 3$, {\sc All Odd Optimal $k$-Vertex Colorings} is still hard in these graphs and consequently so does {\sc Improper Twin $k$-Edge Coloring} if $k\geq 6$, $k \equiv 2 \pmod 4$.  

\begin{prop}\label{prop:bounded_degree}
For all $k\geq 3$, {\sc All Odd Optimal $k$-Vertex Colorings} is coNP-complete in graphs of maximum degree $k+\left\lceil\sqrt{k-2}\hspace{1pt}\right\rceil$.\\
For $k\geq 6$, $k \equiv 2 \pmod 4$ {\sc Improper Twin $k$-Edge Coloring} is NP-complete in graphs of maximum degree $k+\left\lceil\sqrt{k-2}\hspace{1pt}\right\rceil$.	
\end{prop}

\begin{proof}
For the sake of simplicity we denote $g(k)=k+\left\lceil\sqrt{k-2}\hspace{1pt}\right\rceil$.

We show that for $k\geq 3$, {\sc All Odd Optimal $k$-Vertex Colorings} polynomially reduces to the same problem in graphs of maximum degree at most $g(k)$. 
Consider a graph $G$ of maximum degree $\Delta$ as an instance of {\sc All Odd Optimal $k$-Vertex Colorings} for some $k\geq 3$. If $\Delta > g(k)$, then we use the following gadget  to construct a graph $G'$ where the number of vertices of degree~$\Delta$ is one less (eventually 0) than in $G$ and there is no vertex of larger degree, such that $G$ is a $yes$-instance for {\sc All Odd Optimal $k$-Vertex Colorings} if and only if $G'$ is.

 Let $\ell = 1+ \left\lceil\sqrt{k-2}\hspace{1pt}\right\rceil$. Let $p=2\left\lceil \frac{\Delta}{2\ell}\hspace{1pt}\right\rceil+2$. So, $p$ is even and $\ell(p-1)\geq \Delta+1$.
$H(p,\ell)$ is defined as follows:
\begin{enumerate}
\item Consider $p$ independent copies of $(k-1)$-cliques $K_1, \ldots, K_p$.
\item \label{enum:s}Add an independent set $S$ of size $s\geq \Delta$ such that $s$ is odd and $\ell(p-1)-1\leq s\leq \ell(p-1)$. $S=\{x_{ij}, i=1, \ldots, p-1, j=1, \ldots, \ell\}$ if $\ell(p-1)\equiv 1  \pmod 2$ and $S=\{x_{ij}, i=1, \ldots, p-1, j=1, \ldots, \ell\}\setminus \{x_{p-1,\ell}\}$ if $\ell(p-1)\equiv 0  \pmod 2$. 
\item For $i=1, \ldots, p-1$, link by an edge every vertex of $K_i$ with vertices $x_{ij}, j\leq \ell$.
\item For $i=1, \ldots, p-1$, partition $K_{i+1}$ into $\ell$ parts of size at most $\left\lceil \frac{k-1}{\ell}\right\rceil$ and link by an edge vertex $x_{ij}$ with all vertices of part $j$, $j=1, \ldots, \ell$.
\end{enumerate}
Within this construction it is straightforward to verify that in any $k$-vertex coloring of $H(p,\ell)$, all vertices of $S$ are of the same color and each of the other colors is used exactly once in each clique $K_i, i=1, \ldots, p$.

\begin{figure}[ht]
\begin{center}
\includegraphics[scale=0.4]{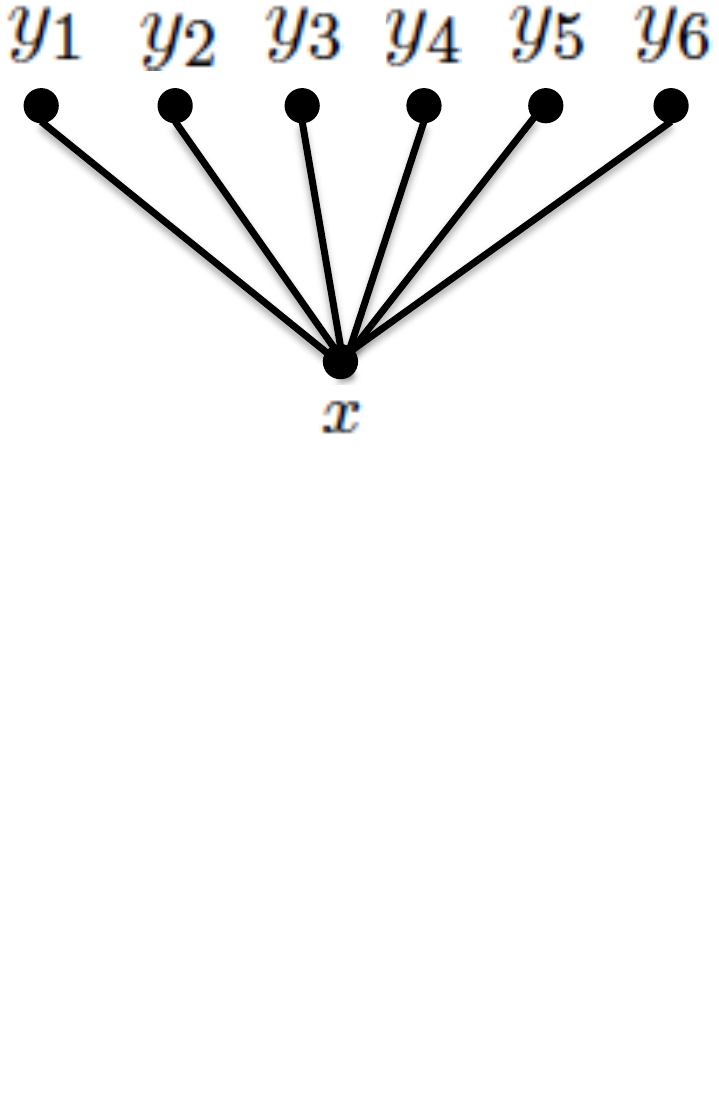}
\includegraphics[scale=0.4]{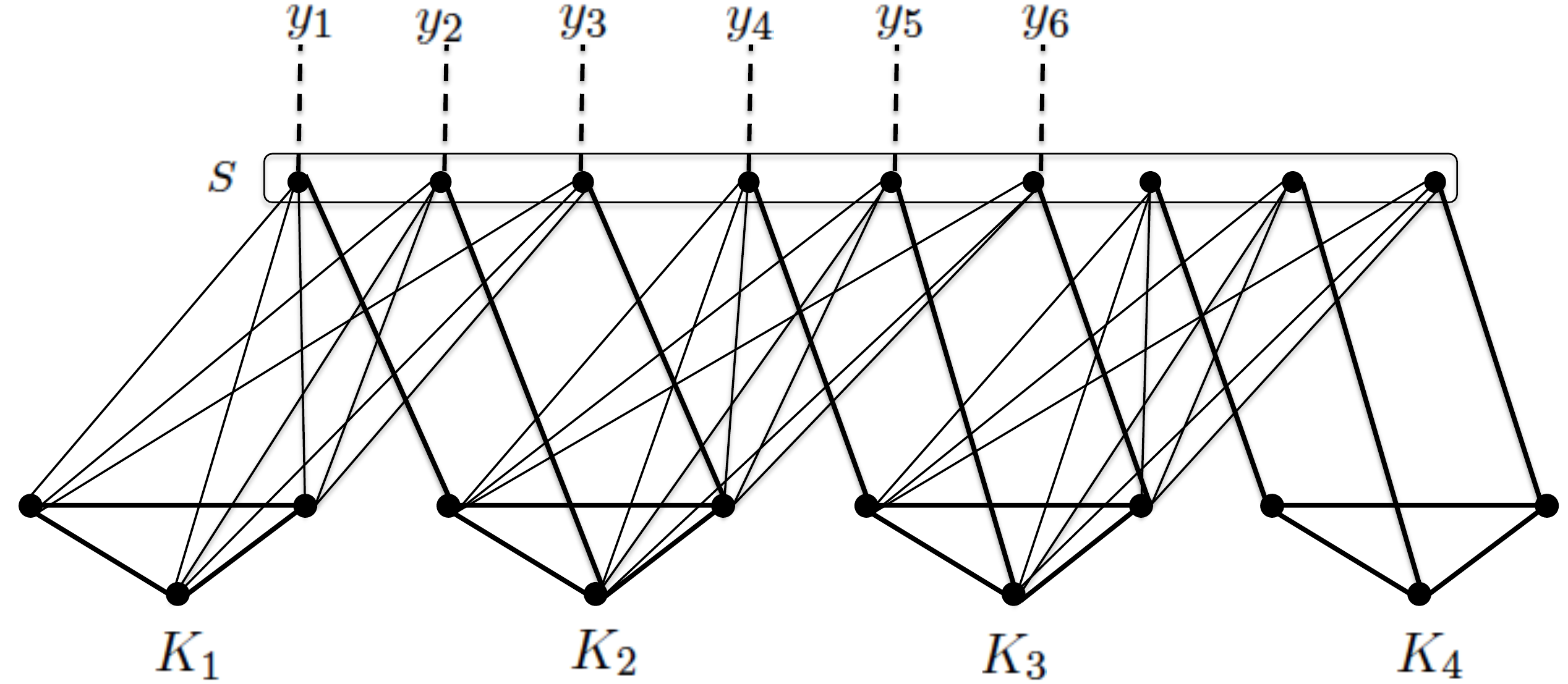}
\caption{The Gadget $H(4,3)$ and how it replaces a vertex $x$ of degree 6. $k=4, \ell=3, p=4$.}\label{fig:bound}
\end{center}
\end{figure}

Given the graph $G$ of maximum degree $\Delta$, if $\Delta > g(k)$, we consider a vertex $x$ of degree $\Delta$ and number its neighbors $y_1, \ldots, y_\Delta$. Then we transform $G$ as follows. We first remove $x$ as well as all its adjacent edges; then we add a copy of $H(p,\ell)$ and link by an edge $y_t$ with  $x_{\lceil \frac{t}{\ell}\rceil ,t-\ell(\lceil \frac{t}{\ell}\rceil -1)}$, $t=1, \ldots, \Delta$, so as to create a matching $M$ of size $\Delta$ between $S$ and the neighborhood of $x$. We denote by $G'$ the resulting graph (see Figure~\ref{fig:bound}). We evaluate the degree of vertices of $H(p,\ell)$ in the graph $G'$: vertices of the cliques $K_i$, $i=1, \ldots, p$ are of degree at most $k-2+\ell+1=g(k)$ and vertices of $S$ are (taking into account one possible edge of $M$) of degree at most  $k-1+\left\lceil \frac{k-1}{\ell}\right\rceil+1\leq k+\left\lceil \frac{k-2+1}{\lceil\sqrt{k-2}\rceil+1}\right\rceil \leq g(k)$. Since $g(k)<\Delta$, $G'$ has one less vertex of degree $\Delta $. 

Using the property of $k$-vertex colorings of $H(p,\ell)$, every $k$-vertex coloring of $G$ with $x$ of color $a$ can be transformed into a $k$-vertex coloring of $G'$ with $S$ colored $a$ and moreover, the number of vertices colored $a$ is increased by $s-1$ and the number of vertices in all other colors is increased by $p$. Conversely every $k$-vertex coloring of $G'$ with $S$ colored $a$ can be transformed into a $k$-vertex coloring of $G$ with $s-1$ vertices less colored $a$ and all other color classes have $p$ less vertices. Since $p$ and $s-1$ are both even, $G$ is a $yes$-instance for {\sc All Odd Optimal $k$-Vertex Colorings} if and only if $G'$ is. 

The construction from $G$ to $G'$ can be performed in polynomial time and by repeating it while there are still vertices of degree greater than $g(k)$ we show that {\sc All Odd Optimal $k$-Vertex Colorings} polynomially reduces to the same problem in graphs of maximum degree $g(k)$. Theorem~\ref{thm:final} allows to conclude the proof.
\end{proof}

Note that the result of Proposition~\ref{prop:bounded_degree} for $k=3$ could be obtained by using the gadgets $H_d$, $d\geq 3$, used in~\cite{gj}, page~86; however this gadget cannot be (easily) generalized for larger $k$. 

In particular, Proposition~\ref{prop:bounded_degree} for $k=3$ and for $k=6$ respectively implies that {\sc All Odd Optimal $3$-Vertex Colorings} is still coNP-complete in graphs of maximum degree~4 and {\sc Improper Twin $6$-Edge Coloring} is NP-complete in  graphs of maximum degree~8. It leaves open the cases with lower degrees. 

\section{Polynomial time solvable cases} \label{sec: poly}

In this section we exhibit some graph classes for which the problems under consideration are polynomial.
We will say that a vertex coloring is {\em  complete} if all colors are maximal independent sets. Equivalently, it means that all colors appear in the closed neighborhood of every vertex.

Let us start with the following remark which will shed light on the forthcoming results.

 \begin{rmk}\label{rem:maximal}\mbox{}\\
	(1) If for a  graph $G$ all color classes have an odd size in all $\chi(G)$-vertex colorings, then all $\chi(G)$-vertex colorings of $G$ are complete;\\
	(2)  All $\chi(G)$-vertex colorings of $G$ are complete if and only if for all connected component $H$ of $G$,  all $\chi(G)$-vertex colorings of $H$ are complete (in particular $\chi(H)=\chi(G)$).\\
	(3) If all $\chi(G)$-vertex colorings of $G$ are complete, then every vertex is of degree at least $\chi(G)-1$.  
\end{rmk}
\begin{proof}
(1) If a graph $G$ has a $\chi(G)$-vertex coloring such that one color class is not a maximal independent set, then if all colors are of odd size, it is possible to transfer one vertex from one color to a non-maximal color class so as to create two color classes of even size.

(2)  It is clear that if one component $H$ has an incomplete $\chi(G)$-vertex coloring (in particular if $\chi(H)<\chi(G)$), then $G$ has an incomplete $\chi(G)$-vertex coloring.

(3) If $G$ has a complete vertex coloring, then the closed neighbourhood of every vertex has at least $k$ vertices of different colors.
\end{proof}

The first polynomial case below, compared to   Proposition~\ref{prop:bounded_degree}, emphasizes that, for bounded degree graphs, there is still a gap of value $O(\sqrt{\chi(G)})$ for the maximum degree  between hard cases and polynomial cases.

\begin{rmk}
{\sc All Odd Optimal $k$-Vertex Colorings} is polynomial in graphs of maximum degree $k-1$.	The graph is a $yes$-instance if and only if it is the disjoint union of an odd number of $k$-cliques.
\end{rmk}

\begin{proof}
	 Consider a graph $G$ of maximum degree $k-1$, with $k=\chi(G)$. We distinguish three kinds of connected components of $G$: $k$-cliques (we denote by $p$ their number), odd cycles of length at least 5 (we denote by $q$ their number) and other components (we denote by $r$ their number). 
	
	Using Brook's theorem~\cite{brook}, all components that are neither a $k$-clique nor an odd cycle of length at least~5 are $(k-1)$-vertex colorable. So, if $r>0$, Remark~\ref{rem:maximal} allows to conclude that $G$ has a $\chi(G)$-vertex coloring with at least one color class of even size.   Odd cycles of size at least~5 are either $(k-1)$-vertex colorable as well if $k\geq 4$ or have a $k$-vertex coloring which is incomplete if $k=3$ (one can ensure that one color is of size~1). In both cases, if $q>0$, Remark~\ref{rem:maximal} also ensures that $G$ has a $\chi(G)$-vertex coloring with at least one color of even size.  
	The only case where all color classes are of odd size in all $k$-vertex colorings is if $r+q=0$ and $p$ is odd.
	 \end{proof}



We recall that a graph is a {\em split graph} if its vertex set can be decomposed into a clique and an independent set. Also, a graph $G$ is called {\em perfect} if every induced subgraph $G'$ of $G$ has a (maximum) clique of size~$\chi(G')$. It is well known (see e.g.~\cite{golumbic}) that a split graph $G$ can be decomposed into a maximum clique $K$  and an independent set $S$. Since split graphs are perfect, we have $|K|=\chi(G)$.

\begin{prop}
	Consider a split graph $G$ decomposed into a maximum clique $K$ of size $k=\chi(G)$ and an independent set $S$ of size $s$. Then all color classes are of odd size in all $k$-vertex colorings if and only if every vertex in the independent set has degree $(k-1)$ and the degree of every vertex in the clique is of the opposite parity as the number of vertices $(k+s)$.
\end{prop}
\begin{proof}
Note first that if $S=\emptyset	$, then the condition holds and in this case all color classes are singletons. Suppose now $S\neq\emptyset	$.
Using Remark~\ref{rem:maximal}, a necessary condition to have all color classes of odd size in all $k$-vertex colorings is that the degree of every vertex in  $S$ is equal to $k-1$ (it cannot be more since $K$ is a maximum clique). 
Then there is a unique $k$-vertex coloring (up to the permutation of colors); each color class has odd size if and only if every vertex in  $K$ is not adjacent to an even number of vertices in $S$. It follows that for every vertex  $x\in K$ the number of adjacent vertices  in $S$ is of the same parity as $s$; moreover, $x$ is adjacent to all other $k-1$ vertices in $K$, this is equivalent to the fact that the degree of $x$ is of the opposite parity as the number of vertices $(k+s)$.
\end{proof}



	


A class of graphs is called {\em hereditary} if whenever a graph belongs to the class all its induced subgraphs also belong to it.

\begin{prop}\label{pro:cochordal}
Consider a hereditary class of graphs $\mathcal{G}$ for which all maximal independent sets and the chromatic number 
 can be computed in polynomial time. Then {\sc All Odd Optimal $k$-Vertex Colorings} and {\sc Minimum Improper Twin Edge Coloring} are polynomial in $\mathcal{G}$.
\end{prop}
\begin{proof}
For any graph $G\in \mathcal{G}$, to decide whether $G$ admits an optimal coloring with at least one color class including an even number of vertices, we proceed in two phases. First enumerate all maximal independent sets with an even number of vertices and test  for each one whether the graph obtained by removing its vertices has chromatic number $\chi(G)-1$. If it is the case for at least one independent set of even size,
 then this independent set is a color of even size in at least one optimal vertex coloring. 
 Else consider all maximal independent sets $S_0$ with odd number of vertices  and for every $x\in S_0$ remove $S_0\setminus \{x\}$ and test whether the remaining graph has chromatic number $\chi(G)-1$. If it is the case for some $S_0$, then there is also an optimal coloring with a color class of even size.

Suppose now that the two phases failed to find  a coloring with an even color class. From the first phase, we know that every complete $\chi(G)$-vertex coloring has all color classes of odd size. Note finally that, if there is an incomplete optimal coloring, then by transferring some vertices from other colors to a non-maximal color class we can construct an optimal coloring where one color is a maximal independent set minus one vertex. This color class cannot be of odd size if the first phase failed.  If the second phase also failed, it cannot be of even size neither, and consequently all  optimal colorings are complete and thus all color classes are of odd size in  all optimal colorings. We conclude by Theorem \ref{thm:final}. \end{proof}

A graph is {\em chordal} (or triangulated) if every cycle of size greater than~3 has a {\it chord}, that is an edge between two non consecutive vertices in the cycle. A graph is co-chordal if its complementary is chordal.
It is known that in chordal graphs all maximal cliques can be computed with complexity $O(|V|+|E|)$~(\cite{golumbic}). Since co-chordal graphs are perfect, Proposition~\ref{pro:cochordal} applies.

\begin{cor}
{\sc All Odd Optimal $k$-Vertex Colorings} and {\sc Minimum Improper Twin Edge Coloring} are polynomial in co-chordal graphs.	
\end{cor}

We finally consider the class of co-comparability graphs that includes for instance interval graphs and permutation graphs (see~\cite{golumbic} for definitions of these classes). A graph is a {\em comparability} graph if it has a transitive orientation of its edges, that is an orientation of edges such that, for every two adjacent edges  $xy$ and $yz$, if $xy$ is oriented from $x$ to $y$ and $yz$ is oriented from $y$ to $z$, then $xz$ is an edge and it is oriented from $x$ to $z$. A graph is a co-comparability graph if its complement is a comparability graph.  A co-comparability graph $G=(V,E)$ has a {\em co-comparability order} that is an order of vertices $x_1, \ldots, x_n$ such that $\forall a<b<c$, if $x_ax_b\notin E$ and $x_bx_c\notin E$, then $x_ax_c\notin E$~\cite{cocomparability}. 

For any $k$-vertex coloring ($k\geq \chi(G)$) we will call the {\em type} of this coloring the list $[n_1, \ldots, n_{k}]$ of the sizes of the different colors (we may have $n_a=0$ for some $a$). Of course $\sum_{a=1}^{k}n_a=|V|$. In order to avoid useless repetitions of types due to the order of colors, we systematically write a type with colors ordered in  non-decreasing order of their size.

\begin{prop}\label{pro:cocomp}
	Let $G$ be a co-comparability graph of bounded chromatic number, then all possible types of $\chi(G)$-vertex colorings can be computed in polynomial time.
\end{prop}

\begin{proof}
We consider a co-comparability order $x_1, \ldots, x_n$ and process by dynamic programming. For a vertex $x_a$, $a$ will be called its {\em rank}. We denote by $x_0$ a fictive vertex considered as isolated. For any $1\leq \ell\leq n$, we denote by $G_\ell$ the subgraph of $G$ induced by $x_1, \ldots, x_\ell$. The number $T$ of  possible types of $\chi(G)$-vertex colorings of $G_\ell$ is bounded by $O(\ell^{\chi(G)})$, which is  polynomial in $n$ since $\chi(G)$ is bounded. For any coloring of type $[n_1, \ldots, n_{\chi(G)}]$ we denote the colors $c_1, \ldots, c_{\chi(G)}$ of size $n_1, \ldots, n_{\chi(G)}$, respectively. We define its {\em state} as the list $[ (n_1,y_1), \ldots, (n_{\chi(G)},y_{\chi(G))})]$, where $y_a$ is the vertex of  maximum rank in color $c_a$ if $n_a\neq 0$ and $y_a=x_0$ if $n_a=0$.  Denoting by $S_\ell$ the number of possible states of  $\chi(G)$-vertex colorings of $G_\ell$, we have $S_\ell\leq  	(\ell+1) T$. 

Note that, using the key property of the co-comparability order, for any independent set $\sigma$ of $G_\ell$, $\ell < n$, with  $x_a$ ($a\leq \ell$) the vertex of maximum rank in $\sigma$, $x_{\ell+1}\cup \sigma$ is independent if and only if $x_{\ell+1}$ and $x_a$ are not adjacent. 
Thus, from all possible states of  $G_\ell$ we can compute in polynomial time all possible states in $G_{\ell+1}$. Starting from the states $[(0,x_0), \ldots, (0, x_0), (1, x_1)]$ we can compute in polynomial time all states of $\chi(G)$-vertex colorings of $G$. The set of all possible types can be trivially deduced.
\end{proof}

The list of types allows immediately to detect all odd colorings and yields the following result using Theorem \ref{thm:final}.

\begin{cor}
	{\sc All Odd Optimal $k$-Vertex Colorings} and {\sc Minimum Improper Twin Edge Coloring} are polynomial in co-comparability graphs with bounded chromatic number.
\end{cor}

These simple examples of polynomial cases allow us to underline some specific situations leading to a polynomial algorithm. These algorithms are all based on an enumeration of  all possible color sizes. Cliques or connected bipartite graphs are uniquely colorable, which makes {\sc All Odd Optimal $k$-Vertex Colorings} trivial. The case of split graphs also reduces to the case of a uniquely colorable graph. Proposition~\ref{pro:cochordal} exhibits a situation where we can exhaustively look for one of the colors while  in the case of Proposition~\ref{pro:cocomp} the sizes of all color classes can be computed in polynomial time. Remark~\ref{rem:maximal} gives a completely different approach that does not require the explicit values of the sizes of color classes. It would be very interesting to find other classes where similar approaches can be applied as well as cases inducing completely new methods. 

In terms of graph classes, these examples underline open cases we leave for future research. Let us mention in particular the complexity in perfect graphs and in planar graphs as well as the case of interval graphs or permutation graphs with unbounded chromatic number. 

\vspace{0.2cm}
\bibliographystyle{abbrvwithurl}
\bibliography{}

\begin{thebibliography}{1}

\bibitem{vertexcol2} L. Addario-Berry, R.E.L. Aldred, K. Dalal, and B.A. Reed, Vertex colouring edge partitions, J. Combin. Theory (B) 94 (2005) 237-244.

\bibitem{and} E. Andrews, L. Helenius, D. Johnston, J. VerWys and Ping Zhang,  On twin edge colorings of graphs, Discussiones Mathematicae Graph Theory 34 (2014) 613-627.

\bibitem{groupsum} M. Anholcer, S. Cichacz, Group sum chromatic number of graphs, European Journal of Combinatorics 55 (2016) 73-81.

\bibitem{4CT} K. Appel, W. Haken, The Solution of the Four-Color Map Problem, Sci. Amer. 237 (1977) 108-121.


\bibitem{brook} R.L. Brooks, On colouring the nodes of a network, 
Mathematical Proceedings of the Cambridge Philosophical Society  37 (1941) 194–197.


\bibitem{chromaticGT} G. Chartrand, P. Zhang, Chromatic Graph theory, CRC Press, 2008.

\bibitem{discrete-notes}
G. Clarke, M. Demange, V. Roshchina, Lecture Notes - Discrete Mathematics, RMIT University.


\bibitem{gj}
M.R. Garey, D.S. Johnson,
Computers and Intractability, a Guide to the Theory of $\mathcal{NP}$-completeness,
Freeman, New York (1979).



\bibitem{golumbic} M.C. Golumbic, Algorithmic Graph Theory and Perfect Graphs, Computer Science and Applied Mathematics, Academic Press, 1980.


\bibitem{JCMCC11} R. Jones, K. Kolasinski, F. Okamoto and P. Zhang, Modular neighbor-distinguishing edge colorings of graphs, J. Combin. Math. Combin. Comput. 76 (2011), 159-175.

\bibitem{JCMCC12} R. Jones, K. Kolasinski, F. Okamoto and P. Zhang, On modular chromatic indexes of graphs, J. Combin. Math. Combin. Comput. 82 (2012), 295-306.

\bibitem{vertexcol1} M. Karonski, T. Luczak, and A. Thomason, Edge weights and vertex colours, J. Combin. Theory (B) 91 (2004) 151-157.

\bibitem{cocomparability}
D. Kratsch, L. Stewart, Domination on cocomparability graphs, SIAM Journal on Discrete Mathematics 6 (1993) 400-417.

\bibitem{planar} N. Robertson, D.P. Sanders, P. Seymour and R. Thomas, Efficiently four-coloring planar
graphs, In Proceedings of the twenty-eighth annual ACM symposium on Theory of computing, 571-575 ACM, 1996.

\bibitem{colorinduced} P. Zhang, Color-Induced Graph Colorings, Springer, 2015.

\end{thebibliography}



\end{document}